\documentclass[twocolumn,english]{IEEEtran}
\usepackage[T1]{fontenc}
\usepackage[latin1]{inputenc}
\usepackage{amsmath}
\usepackage{graphicx}
\usepackage{amssymb}

\makeatletter
  \newtheorem{thm}{Theorem}
  \newtheorem{cor}{Corollary}
  \newtheorem{remrk}{Remark}

\usepackage{cite}
\usepackage{url}
\usepackage{bm}
\usepackage{bbm}

\usepackage{babel}
\makeatother
\begin{document}

\title{Unequal dimensional small balls and quantization on Grassmann Manifolds$^{^{*}}$}

\author{Wei Dai$^{\dagger}$, Brian Rider$^{\dagger\dagger}$ and Youjian(Eugene)
Liu$^{\dagger}$\\
$^{\dagger}$Department of Electrical and Computer Engineering, $^{\dagger\dagger}$Department
of Mathematics\\
 University of Colorado at Boulder\\
\{wei.dai, brian.rider\}@colorado.edu, eugeneliu@ieee.org}

\maketitle

\begin{abstract}
The Grassmann manifold $\mathcal{G}_{n,p}\left(\mathbb{L}\right)$
is the set of all $p$-dimensional planes (through the origin) in
the $n$-dimensional Euclidean space $\mathbb{L}^{n}$, where $\mathbb{L}$
is either $\mathbb{R}$ or $\mathbb{C}$. This paper considers an
unequal dimensional quantization in which a source in $\mathcal{G}_{n,p}\left(\mathbb{L}\right)$
is quantized through a code in $\mathcal{G}_{n,q}\left(\mathbb{L}\right)$,
where $p$ and $q$ are not necessarily the same. It is different
from most works in literature where $p\equiv q$. The analysis for
unequal dimensional quantization is based on the volume of a metric
ball in $\mathcal{G}_{n,p}\left(\mathbb{L}\right)$ whose center is
in $\mathcal{G}_{n,q}\left(\mathbb{L}\right)$. Our chief result is
a closed-form formula for the volume of a metric ball when the radius
is sufficiently small. This volume formula holds for Grassmann manifolds
with arbitrary $n$, $p$, $q$ and $\mathbb{L}$, while previous
results pertained only to some special cases. Based on this volume
formula, several bounds are derived for the rate distortion tradeoff
assuming the quantization rate is sufficiently high. The lower and
upper bounds on the distortion rate function are asymptotically identical,
and so precisely quantify the asymptotic rate distortion tradeoff.
We also show that random codes are asymptotically optimal in the sense
that they achieve the minimum achievable distortion with probability
one as $n$ and the code rate approach infinity linearly.

Finally, we discuss some applications of the derived results to communication
theory. A geometric interpretation in the Grassmann manifold is developed
for capacity calculation of additive white Gaussian noise channel.
Further, the derived distortion rate function is beneficial to characterizing
the effect of beamforming matrix selection in multi-antenna communications.
\end{abstract}
\begin{keywords}
the Grassmann manifold, rate distortion tradeoff, channel capacity,
beamforming, MIMO communications
\end{keywords}
\renewcommand{\thefootnote}{\fnsymbol{footnote}} \footnotetext[1]{This work is supported by the NSF DMS-0508680, Thomson Inc., and  the Junior Faculty Development Award, University of Colorado at Boulder. Part of this work was published in  \cite{Dai_Globecom05_Quantization_bounds_Grassmann_manifold}. This paper extends the quantization in  \cite{Dai_Globecom05_Quantization_bounds_Grassmann_manifold} to the  unequal dimensional case.} \renewcommand{\thefootnote}{\arabic{footnote}} \setcounter{footnote}{0}

\section{Introduction\label{sec:Introduction}}


The \emph{Grassmann manifold} $\mathcal{G}_{n,p}\left(\mathbb{L}\right)$
is the set of all $p$-dimensional planes (through the origin) in
the $n$-dimensional Euclidean space $\mathbb{L}^{n}$, where $\mathbb{L}$
is either $\mathbb{R}$ or $\mathbb{C}$. It forms a compact Riemann
manifold of real dimension $\beta p\left(n-p\right)$, where $\beta=1$
when $\mathbb{L}=\mathbb{R}$ and $\beta=2$ when $\mathbb{L}=\mathbb{C}$.
The Grassmann manifold is a useful analysis tool for multi-antenna
communications (also known as multiple-input multiple-output (MIMO)
communication systems). The capacity of non-coherent MIMO systems
at high signal-to-noise ratio (SNR) region was derived by analysis
in the Grassmann manifold \cite{Tse_IT02_Communications_on_Grassmann_manifold}.
The well known spherical codes for MIMO systems can be viewed as codes
in the Grassmann manifold \cite{Urbanke_IT01_Signal_Constellations}.
Further, for coherent MIMO systems with finite rate feedback, the
quantization of eigen-channel vectors is related to the quantization
on the Grassmann manifold \cite{Sabharwal_IT03_Beamforming_MIMO,Love_IT03_Grassman_Beamforming_MIMO,Dai_05_Power_onoff_strategy_design_finite_rate_feedback,Rao_SP07_Feedback_High_Resolution,Jindal_ISIT2006_Feedback_Reduction_MIMO_BC}. 

This paper studies unequal dimensional quantization on the Grassmann
manifold. Roughly speaking, a quantization is a representation of
a source: it maps an element in $\mathcal{G}_{n,p}\left(\mathbb{L}\right)$
(the source) into a subset $\mathcal{C}\subset\mathcal{G}_{n,q}\left(\mathbb{L}\right)$,
which is often discrete and referred to as a \emph{code}. While it
is traditionally assumed that $p\equiv q$ \cite{Conway_96_PackingLinesPlanes,Barg_IT02_Bounds_Grassmann_Manifold,Henkel_IT2005_Sphere_Packing_Bounds_Grassmann_Stiefel,Dai_Globecom05_Quantization_bounds_Grassmann_manifold},
we are interested in a more general case where $p$ may not necessarily
equal to $q$; thus the term unequal dimensional quantization. The
performance limit of quantization is given by the so called rate distortion
tradeoff. Let the source be randomly distributed and define a \emph{distortion}
metric between elements in $\mathcal{G}_{n,p}\left(\mathbb{L}\right)$
and $\mathcal{G}_{n,q}\left(\mathbb{L}\right)$. The rate distortion
tradeoff is described by the minimum average distortion achievable
for a given code size, or equivalently the minimum code size required
to achieve a particular average distortion. This paper will quantify
the rate distortion tradeoff for unequal dimensional quantization.

This paper appears to be the first to explore  unequal dimensional
quantization systematically. According to the authors' knowledge,
works in literature assume that $p=q$: The Rankin bound in $\mathcal{G}_{n,p}\left(\mathbb{R}\right)$
is obtained in \cite{Conway_96_PackingLinesPlanes} when the code
size is large. When $p$ is fixed and $n$ is asymptotically large,
approximations to the Gilbert-Varshamov and Hamming bounds on $\mathcal{G}_{n,p}\left(\mathbb{L}\right)$
are drived by Laplace method in \cite{Barg_IT02_Bounds_Grassmann_Manifold}
and by volume estimates in \cite{Henkel_IT2005_Sphere_Packing_Bounds_Grassmann_Stiefel,Han_IT2006_Unitary_Space_Time_Diversity_Analysis}.
The distortion rate tradeoff for the $p=1$ case is quantified in
\cite{Sabharwal_IT03_Beamforming_MIMO,Love_IT03_Grassman_Beamforming_MIMO}
by direct volume calculation and in \cite{Rao_SP07_Feedback_High_Resolution}
using high resolution quantization theory. Our paper \cite{Dai_Globecom05_Quantization_bounds_Grassmann_manifold}
characterizes the tradeoff for the general $p$ case when quantization
rate is sufficiently high. While the $p=q$ case has been extensively
studied, unequal dimensional quantization does arise in some multi-antenna
communication systems, see \cite{Jindal_ISIT2006_Feedback_Reduction_MIMO_BC}
for an example. It is thus worthwhile to go beyond the $p=q$ case.

The main contribution of this paper is to derive a closed-form formula
for the volume of a small ball in the Grassmann manifold and then
accurately quantify the rate distortion tradeoff accordingly. Specifically:

\begin{enumerate}
\item An explicit volume formula for a metric ball is derived for arbitrary
$n$, $p$, $q$ and $\mathbb{L}$ when the radius $\delta$ is sufficiently
small. Useful lower and upper bounds on the volume are also presented. 
\item Tight lower and upper bounds are derived for the rate distortion tradeoff.
Further, fix $p$ and $q$ but let $n$ and the code rate (logarithm
of the code size) approach infinity linearly. The lower and upper
bounds are in fact asymptotically identical, and so precisely quantify
the asymptotic rate distortion tradeoff. We also show that random
codes are asymptotically optimal in the sense that they achieve the
minimum achievable distortion with probability one in this  asymptotic
region.
\end{enumerate}
Finally, some applications of the derived results to communication
theory are presented. We show that data transmission in additive white
Gaussian noise (AWGN) channel is essentially communication on the
Grassmann manifold. A geometric interpretation for AWGN channel is
developed in the Grassmann manifold accordingly. Moreover, the beamforming
matrix selection in a MIMO system is closely related to quantization
on the Grassmann manifold. The results for the distortion rate tradeoff
are therefore helpful to characterize the effect of beamforming matrix
selection.

\section{\label{sec:Preliminaries}Preliminaries}

For the sake of applications \cite{Sabharwal_IT03_Beamforming_MIMO,Love_IT03_Grassman_Beamforming_MIMO,Dai_05_Power_onoff_strategy_design_finite_rate_feedback},
the projection Frobenius metric (\emph{chordal distance}) and the
invariant measure on the Grassmann manifold are employed throughout
this paper. Without loss of generality, we assume that $p\le q$.
For any two planes $P\in\mathcal{G}_{n,p}\left(\mathbb{L}\right)$
and $Q\in\mathcal{G}_{n,q}\left(\mathbb{L}\right)$, we define the
principle angles and the chordal distance between $P$ and $Q$ as
follows. Let $\mathbf{p}_{1}\in P$ and $\mathbf{q}_{1}\in Q$ be
the unit vectors such that $\left|\mathbf{p}_{1}^{\dagger}\mathbf{q}_{1}\right|$
is maximal. Inductively, let $\mathbf{p}_{i}\in P$ and $\mathbf{q}_{i}\in Q$
be the unit vectors such that $\mathbf{p}_{i}^{\dagger}\mathbf{p}_{j}=0$
and $\mathbf{q}_{i}^{\dagger}\mathbf{q}_{j}=0$ for all $1\leq j<i$
and $\left|\mathbf{p}_{i}^{\dagger}\mathbf{q}_{i}\right|$ is maximal.
The principle angles are then defined as $\theta_{i}=\arccos\left|\mathbf{p}_{i}^{\dagger}\mathbf{q}_{i}\right|$
for $i=1,\cdots,p$ \cite{Conway_96_PackingLinesPlanes}, and the
chordal distance between $P$ and $Q$ is then given by \begin{equation}
d_{c}\left(P,Q\right)\triangleq\sqrt{\sum_{i=1}^{p}\sin^{2}\theta_{i}}.\label{eq:def-dc}\end{equation}
The invariant measure $\mu$ on $\mathcal{G}_{n,p}\left(\mathbb{L}\right)$
is the Haar measure on $\mathcal{G}_{n,p}\left(\mathbb{L}\right)$.
Let $O\left(n\right)$ and $U\left(n\right)$ be the groups of $n\times n$
orthogonal and unitary matrices respectively. Let $\mathbf{A},\mathbf{B}\in O\left(n\right)$
when $\mathbb{L}=\mathbb{R}$, or $\mathbf{A},\mathbf{B}\in U\left(n\right)$
when $\mathbb{L}=\mathbb{C}$. For any measurable set $\mathcal{M}\subset\mathcal{G}_{n,p}\left(\mathbb{L}\right)$
and arbitrary $\mathbf{A}$ and $\mathbf{B}$, $\mu$ satisfies \[
\mu\left(\mathbf{A}\mathcal{M}\right)=\mu\left(\mathcal{M}\right)=\mu\left(\mathcal{M}\mathbf{B}\right).\]
The invariant measure defines the uniform/isotropic distribution on
$\mathcal{G}_{n,p}\left(\mathbb{L}\right)$ \cite{James_54_Normal_Multivariate_Analysis_Orthogonal_Group}.

This paper addresses an unequal dimensional quantization problem.
Let $\mathcal{C}$ be a finite size discrete subset of $\mathcal{G}_{n,q}\left(\mathbb{L}\right)$
(also known as a code). An unequal dimensional quantization is a mapping
from the $\mathcal{G}_{n,p}\left(\mathbb{L}\right)$ to the set $\mathcal{C}$,
$\mathfrak{q}:\mathcal{G}_{n,p}\left(\mathbb{L}\right)\rightarrow\mathcal{C}$,
where $p$ and $q$ are not necessarily the same integer. Without
loss of generality, we assume $p\le q$. We are interested in quantifying
the \emph{rate distortion tradeoff}. Assume that a source $P\in\mathcal{G}_{n,p}\left(\mathbb{L}\right)$
is isotropically distributed. Define the distortion measure as the
square of the chordal distance $d_{c}^{2}\left(\cdot,\cdot\right)$.
Then the distortion associated with a quantization $\mathfrak{q}$
is defined as $D\triangleq\mathrm{E}_{P}\left[d_{c}^{2}\left(P,\mathfrak{q}\left(P\right)\right)\right].$
For a given code $\mathcal{C}\subset\mathcal{G}_{n,q}\left(\mathbb{L}\right)$,
the optimal quantization to minimize the distortion is given by $\mathfrak{q}\left(P\right)=\arg\;\underset{Q\in\mathcal{C}}{\min}\; d_{c}\left(P,Q\right).$
The corresponding distortion is \[
D\left(\mathcal{C}\right)=\mathrm{E}_{P}\left[\underset{Q\in\mathcal{C}}{\min}\; d_{c}^{2}\left(P,Q\right)\right].\]
 The rate distortion tradeoff can be described by the \emph{distortion
rate function}: the infimum achievable distortion given a code size
$K$ \begin{equation}
D^{*}\left(K\right)=\underset{\mathcal{C}:\left|\mathcal{C}\right|=K}{\inf}\; D\left(\mathcal{C}\right),\label{eq:distor-rate-fn}\end{equation}
or the \emph{rate distortion function}: the minimum required code
size to achieve a given distortion $D$\begin{equation}
K^{*}\left(D\right)=\underset{D\left(\mathcal{C}\right)\leq D}{\inf}\;\left|\mathcal{C}\right|.\label{eq:rate-distor-fn}\end{equation}

\section{\label{sec:Spheres}Metric Balls in the Grassmann Manifold}

This section derives an explicit volume formula for a metric ball
$B\left(\delta\right)$ in the Grassmann manifold. It is the essential
tool to quantify the rate distortion tradeoff.

The volume of a ball can be expressed as a multivariate integral.
Assume the invariant measure $\mu$ and the chordal distance $d_{c}$.
For any given $P\in\mathcal{G}_{n,p}\left(\mathbb{L}\right)$ and
$Q\in\mathcal{G}_{n,q}\left(\mathbb{L}\right)$, define \[
B_{P}\left(\delta\right)=\left\{ \hat{Q}\in\mathcal{G}_{n,q}\left(\mathbb{L}\right):\; d_{c}\left(P,\hat{Q}\right)\leq\delta\right\} \]
and \[
B_{Q}\left(\delta\right)=\left\{ \hat{P}\in\mathcal{G}_{n,p}\left(\mathbb{L}\right):\; d_{c}\left(\hat{P},Q\right)\leq\delta\right\} .\]
It has been shown that $\mu\left(B_{P}\left(\delta\right)\right)=\mu\left(B_{Q}\left(\delta\right)\right)$
and the value is independent of the choice of the center \cite{James_54_Normal_Multivariate_Analysis_Orthogonal_Group}.
For convenience, we denote $B_{P}\left(\delta\right)$ and $B_{Q}\left(\delta\right)$
by $B\left(\delta\right)$ without distinguishing them. Then, the
volume of a metric ball $B\left(\delta\right)$ is given by\begin{equation}
\mu\left(B\left(\delta\right)\right)=\underset{\sum_{i=1}^{p}\sin^{2}\theta_{i}\leq\delta^{2}}{\int\cdots\int}\; d\mu_{\bm{\theta}},\label{eq:actual_volume}\end{equation}
where $1\leq\theta_{1}\leq\frac{\pi}{2},\cdots,1\leq\theta_{p}\leq\frac{\pi}{2}$
are the principle angles and the differential form $d\mu_{\bm{\theta}}$
is the joint density of the $\theta_{i}$'s \cite{James_54_Normal_Multivariate_Analysis_Orthogonal_Group,Adler_2004_Integrals_Grassmann}. 

Theorem \ref{thm:Volume_formula} computes the multivariate integral
(\ref{eq:actual_volume}) into a simple exponential form.

\begin{thm}
\label{thm:Volume_formula}When $\delta\leq1$, the volume of a metric
ball $B\left(\delta\right)$ is given by \begin{equation}
\mu\left(B\left(\delta\right)\right)=c_{n,p,q,\beta}\delta^{\beta p\left(n-q\right)}\left(1+c_{n,p,q,\beta}^{\left(1\right)}\delta^{2}+o\left(\delta^{2}\right)\right),\label{eq:simplified_volume_formula}\end{equation}
where\[
\beta=\left\{ \begin{array}{ll}
1 & \mathrm{if}\;\mathbb{L}=\mathbb{R}\\
2 & \mathrm{if}\;\mathbb{L}=\mathbb{C}\end{array}\right.,\]
\begin{equation}
c_{n,p,q,\beta}=\left\{ \begin{array}{l}
\frac{1}{\Gamma\left(\frac{\beta}{2}p\left(n-q\right)+1\right)}\prod_{i=1}^{p}\frac{\Gamma\left(\frac{\beta}{2}\left(n-i+1\right)\right)}{\Gamma\left(\frac{\beta}{2}\left(q-i+1\right)\right)}\\
\quad\quad\quad\quad\quad\quad\quad\quad\quad\quad\;\mathrm{if}\; p+q\le n\\
\frac{1}{\Gamma\left(\frac{\beta}{2}p\left(n-q\right)+1\right)}\prod_{i=1}^{n-q}\frac{\Gamma\left(\frac{\beta}{2}\left(n-i+1\right)\right)}{\Gamma\left(\frac{\beta}{2}\left(n-p-i+1\right)\right)}\\
\quad\quad\quad\quad\quad\quad\quad\quad\quad\quad\;\mathrm{if}\; p+q\ge n\end{array}\right.,\label{eq:constant-0}\end{equation}
and \begin{equation}
c_{n,p,q,\beta}^{\left(1\right)}=-\left(\frac{\beta}{2}\left(q-p+1\right)-1\right)\frac{\frac{\beta}{2}p\left(n-q\right)}{\frac{\beta}{2}p\left(n-q\right)+1}.\label{eq:constant-1}\end{equation}

\end{thm}

The proof is given in the journal version of this paper \cite{Dai_05_Quantization_Grassmannian_manifold}.


There are two cases where the volume formula becomes exact.

\begin{cor}
\label{cor:Exact-Volume-Formula}When $\delta\leq1$, in either of
the following two cases, 
\begin{enumerate}
\item $\mathbb{L}=\mathbb{C}$ and $q=p$;
\item $\mathbb{L}=\mathbb{R}$ and $q=p+1$, 
\end{enumerate}
the volume of a metric ball $B\left(\delta\right)$ is exactly \[
\mu\left(B\left(\delta\right)\right)=c_{n,p,q,\beta}\delta^{\beta p\left(n-q\right)},\]
 where $c_{n,p,q,\beta}$ is defined in (\ref{eq:constant-0}). 
\end{cor}

We also have the general bounds:

\begin{cor}
\label{cor:volume-bounds}Assume $\delta\leq1$. If $\mathbb{L}=\mathbb{R}$
and $p=q$ , the volume of $B\left(\delta\right)$ is bounded by \[
c_{n,p,p,1}\delta^{p\left(n-p\right)}\leq\mu\left(B\left(\delta\right)\right)\leq c_{n,p,p,1}\delta^{p\left(n-p\right)}\left(1-\delta^{2}\right)^{-\frac{p}{2}}.\]
For all other cases, \begin{align*}
 & \left(1-\delta^{2}\right)^{\frac{\beta}{2}p\left(q-p+1\right)-p}c_{n,p,q,\beta}\delta^{\beta p\left(n-q\right)}\\
 & \quad\quad\leq\mu\left(B\left(\delta\right)\right)\leq c_{n,p,q,\beta}\delta^{\beta p\left(n-q\right)}.\end{align*}

\end{cor}
\begin{proof}
Corollary \ref{cor:Exact-Volume-Formula} and \ref{cor:volume-bounds}
follow the proof of Theorem \ref{thm:Volume_formula} by tracking
the higher order terms. 
\end{proof}

Theorem \ref{thm:Volume_formula} is of course consistent with the
previous results in \cite{Dai_Globecom05_Quantization_bounds_Grassmann_manifold,Barg_IT02_Bounds_Grassmann_Manifold,Sabharwal_IT03_Beamforming_MIMO},
which pertain to special choices of $n$, $p$, $q$ or $\mathbb{L}$.
Importantly though, Theorem \ref{thm:Volume_formula} is distinct
in that it holds for arbitrary $n,$ $p$, $q$ and $\mathbb{L}$. 

For engineering purposes, it is often satisfactory to approximate
the volume of a metric ball $B\left(\delta\right)$ by $c_{n,p,q,\beta}\delta^{\beta p\left(n-q\right)}$
when $\delta\le1$. Fig. \ref{cap:volume_in_Grassmann} compares the
simulated volume (\ref{eq:actual_volume}) and the approximation $c_{n,p,q,\beta}\delta^{\beta p\left(n-q\right)}$.
Since it is often difficult to directly evaluate the multivariate
integral in (\ref{eq:actual_volume}), we simulate $\mu\left(B\left(\delta\right)\right)=\Pr\left\{ \hat{P}\in\mathcal{G}_{n,p}\left(\mathbb{L}\right):\; d_{c}\left(\hat{P},Q\right)\leq\delta\right\} $
by fixing $Q$ and generating isotropically distributed $\hat{P}$.
The simulation results show that our volume approximation $c_{n,p,q,\beta}\delta^{\beta p\left(n-q\right)}$
(solid lines) is close to the simulated volume (circles) when $\delta\le1$.
We also compare our approximation with Barg-Nogin approximation developed
in \cite{Barg_IT02_Bounds_Grassmann_Manifold}. There, an volume approximation
$\left(\delta/\sqrt{p}\right)^{\beta np}$ is derived by Laplace method
and is only valid for the $p=q\ll n$. Simulations show that the simulated
volume and Barg-Nogin approximation (dash lines) may not be of the
same order while our approximation is much more accurate.

\begin{figure}[h]
\includegraphics[clip,scale=0.65]{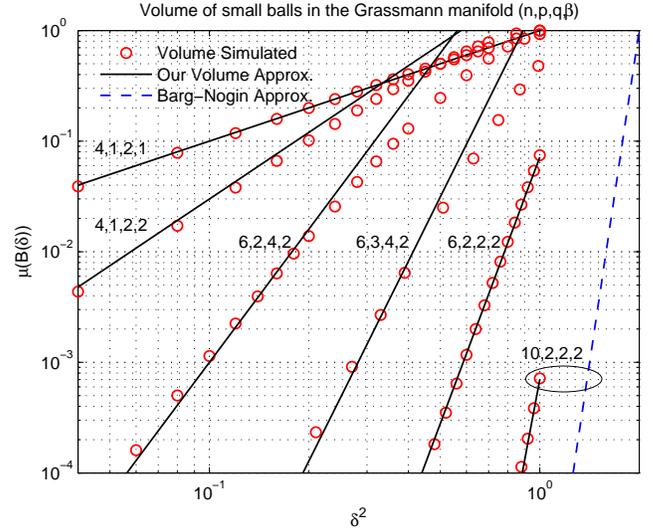}

\caption{\label{cap:volume_in_Grassmann}Volume of small balls in the Grassmann
manifold. The integers besides curves are, from left to right, $n,\; p,\; q,\;\mathrm{and}\;\beta$
respectively.}
\end{figure}

\section{Quantization Bounds\label{sec:Quantization-Bounds}}

This section quantifies the rate distortion tradeoff for the unequal
dimensional quantization problem. The results hold for arbitrary $n$,
$p$, $q$ and $\mathbb{L}$.

Recall the distortion rate function defined in (\ref{eq:distor-rate-fn}).
A lower bound and an upper bound are derived.

\begin{thm}
\label{thm:DRF_bounds}When $K$ is sufficiently large ($\left(c_{n,p,q,\beta}K\right)^{-\frac{2}{\beta p\left(n-q\right)}}\leq1$
necessarily), the distortion rate function is bounded as in \begin{align}
 & \frac{\beta p\left(n-q\right)}{\beta p\left(n-q\right)+2}\left(c_{n,p,q,\beta}K\right)^{-\frac{2}{\beta p\left(n-q\right)}}\left(1+o\left(1\right)\right)\leq D^{*}\left(K\right)\nonumber \\
 & \quad\quad\quad\leq\frac{2\Gamma\left(\frac{2}{\beta p\left(n-q\right)}\right)}{\beta p\left(n-q\right)}\left(c_{n,p,q,\beta}K\right)^{-\frac{2}{\beta p\left(n-q\right)}}\left(1+o\left(1\right)\right).\label{eq:DRF_bounds}\end{align}

\end{thm}
\begin{remrk}
For engineering purposes, the main order terms in (\ref{eq:DRF_bounds})
are usually accurate enough to characterize the distortion rate function.
The details of the $\left(1+o\left(1\right)\right)$ correction are
spelled out in the journal version of this paper \cite{Dai_05_Quantization_Grassmannian_manifold}.
\end{remrk}

The proof is provided in the journal version of this paper \cite{Dai_05_Quantization_Grassmannian_manifold}.
We sketch it as follows.

The lower bound is proved by a sphere covering argument. The key is
to construct an ideal quantizer, which may not exist, to minimize
the distortion. Suppose that there exists $K$ metric balls of the
same radius $\delta_{0}$ packing and covering the whole $\mathcal{G}_{n,p}\left(\mathbb{L}\right)$
at the same time. Then the quantizer which maps each of those balls
into its center $Q\in\mathcal{G}_{n,q}\left(\mathbb{L}\right)$ gives
the minimum distortion among all quantizers. Of course such an ideal
covering  may not exist. Therefore, the corresponding distortion may
not be achievable. It is only a lower bound on the distortion rate
function.

Next the upper bound is obtained by calculating the average distortion
of random codes. The basic idea is that the distortion of any particular
code is an upper bound of the distortion rate function and so is the
average distortion of random codes. A random code $\mathcal{C}_{\mathrm{rand}}=\left\{ Q_{1},\cdots,Q_{K}\right\} $
is generated by drawing the codewords $Q_{i}$'s independently from
the isotropic distribution on $\mathcal{G}_{n,q}\left(\mathbb{L}\right)$.
The average distortion of random codes is given by $\mathrm{E}_{\mathcal{C}_{\mathrm{rand}}}\left[D\left(\mathcal{C}_{\mathrm{rand}}\right)\right]$.
By extreme order statistics, see for example \cite{JanosGalambos1987_extreme_order_statistics},
the calculation of $\mathrm{E}_{\mathcal{C}_{\mathrm{rand}}}\left[D\left(\mathcal{C}_{\mathrm{rand}}\right)\right]$
is directly related to the volume (\ref{eq:actual_volume}). Based
on our volume formula (\ref{eq:simplified_volume_formula}), the asymptotic
value of $\mathrm{E}_{\mathcal{C}_{\mathrm{rand}}}\left[D\left(\mathcal{C}_{\mathrm{rand}}\right)\right]$
is computed and thus the upper bound is obtained for large $K$. 

As the dual part of the distortion rate tradeoff, lower and upper
bounds are constructed for the rate distortion function. 

\begin{cor}
\label{cor:RDF_bounds}When the required distortion $D$ is sufficiently
small ($D\leq1$ necessarily), the rate distortion function satisfies
the following bounds,\begin{align}
 & \frac{1}{c_{n,p,q,\beta}}\left(\frac{\beta p\left(n-q\right)}{2\Gamma\left(\frac{2}{\beta p\left(n-q\right)}\right)}D\right)^{-\frac{\beta p\left(n-q\right)}{2}}\left(1+o\left(1\right)\right)\leq K^{*}\left(D\right)\nonumber \\
 & \quad\quad\leq\frac{1}{c_{n,p,q,\beta}}\left(\frac{\beta p\left(n-q\right)+2}{\beta p\left(n-q\right)}D\right)^{-\frac{\beta p\left(n-q\right)}{2}}\left(1+o\left(1\right)\right).\label{eq:RDF_bounds}\end{align}

\end{cor}

It is interesting to observe that the lower and upper bounds are asymptotically
the same. As a result, the asymptotic rate distortion tradeoff is
exactly quantified. 

\begin{thm}
\label{thm:Asymptotic-Meet}Suppose that $p$ and $q$ are fixed.
Let $n$ and the code rate $\log_{2}K$ approach infinity linearly
with $\frac{\log_{2}K}{n}\rightarrow\bar{r}$. If the normalized code
rate $\bar{r}$ is sufficiently large ($p2^{-\frac{2}{\beta p}\bar{r}}\leq1$
necessarily), then \[
\underset{\left(n,K\right)\rightarrow+\infty}{\lim}D^{*}\left(K\right)=p2^{-\frac{2}{\beta p}\bar{r}}.\]
On the other hand, if the required distortion $D$ is sufficiently
small ($D\leq1$ necessarily), then the minimum code size required
to achieve the distortion $D$ satisfies \begin{equation}
\underset{n\rightarrow+\infty}{\lim}\frac{\log_{2}K^{*}\left(D\right)}{n}=\frac{\beta p}{2}\log_{2}\left(\frac{p}{D}\right).\label{eq:RDF-asymptotics}\end{equation}

\end{thm}
\begin{remrk}
That the $\left(1+o\left(1\right)\right)$ multiplicative errors in
(\ref{eq:DRF_bounds}) and (\ref{eq:RDF_bounds}) disappear is the
content of \cite{Dai_05_Quantization_Grassmannian_manifold}. We omit
the corresponding details due to the space limitation.
\end{remrk}

Fig. \ref{cap:DRF_bounds} compares the simulated distortion rate
function (the plus markers) with its lower bound (the dashed lines)
and upper bound (the solid lines) in (\ref{eq:DRF_bounds}). To simulate
the distortion rate function, we use the max-min criterion to design
codes and the employ the corresponding distortion as an estimate of
the distortion rate function. Simulation results show that the bounds
in (\ref{eq:DRF_bounds}) hold for large $K$. When $K$ is relatively
small, the formula (\ref{eq:DRF_bounds}) can serve as good approximations
to the distortion rate function as well. Furthermore, we compare our
bounds with the approximation (the {}``x'' markers) derived in \cite{Heath_ICASSP05_Quantization_Grassmann_Manifold},
which is partly based on Barg-Nogin volume approximation. Simulations
show that the approximation in \cite{Heath_ICASSP05_Quantization_Grassmann_Manifold}
is neither an upper bound nor a lower bound. It works for the case
that $n=10$ and $p=2$ but doesn't work when $n\leq8$ and $p=2$.
As a comparison, our bounds (\ref{eq:DRF_bounds}) hold for arbitrary
$n$ and $p$.

\begin{figure}
\begin{centering}
\includegraphics[clip,scale=0.6]{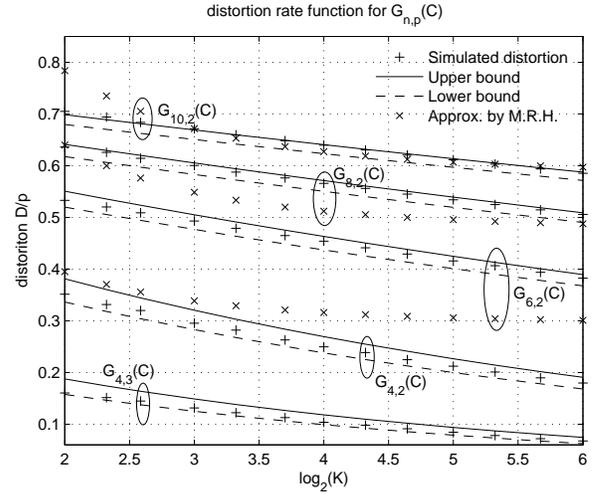}
\par\end{centering}

\caption{\label{cap:DRF_bounds}Bounds on the distortion rate function}
\end{figure}

While the asymptotic rate distortion tradeoff is precisely quantified,
the next question could be how to achieve it. Same to  many cases
in information theory, \emph{random codes are asymptotically optimal
with probability one}.

\begin{cor}
Consider  unequal dimensional quantization from $\mathcal{G}_{n,p}\left(\mathbb{L}\right)$
to $\mathcal{G}_{n,q}\left(\mathbb{L}\right)$. Let $\mathcal{C}_{\mathrm{rand}}\subset\mathcal{G}_{n,q}\left(\mathbb{L}\right)$
be a code randomly generated from the isotropic distribution and with
size $K$. Fix $p$ and $q$. Let $n,\log_{2}K\rightarrow\infty$
with $\frac{\log_{2}K}{n}\rightarrow\bar{r}\in\mathbb{R}^{+}$. If
the normalized code rate $\bar{r}$ is sufficiently large ($p2^{-\frac{2}{\beta p}\bar{r}}\leq1$
necessarily), then for $\forall\epsilon>0$, \[
\underset{\left(n,K\right)\rightarrow\infty}{\lim}\Pr\left(D\left(\mathcal{C}_{\mathrm{rand}}\right)>p2^{-\frac{2}{\beta p}\bar{r}}+\epsilon\right)=0\]

\end{cor}

The proof is omitted due to the space limitation.

\section{\label{sec:Application}Applications to Communication Theory}

\subsection{Channel Capacity of AWGN Channel}

Although the capacity of AWGN channel is well known, it is interesting
to re-calculate it from an interpretation in the Grassmann manifold. 

The signal transmission model for an AWGN channel is that $\mathbf{Y}=\mathbf{X}+\mathbf{W},$
where $\mathbf{Y},\;\mathbf{X},\;\mathbf{W}\in\mathbb{L}^{n}$ are
the received signal, the transmitted signal and the additive Gaussian
noise respectively, and $\mathbb{L}$ is either $\mathbb{R}$ or $\mathbb{C}$.
Assume that $\mathbf{X}$ and $\mathbf{W}$ are Gaussian vectors with
zero mean and covariance matrices $\mathrm{E}\left[\mathbf{XX}^{\dagger}\right]=\mathbf{I}$
and $\mathrm{E}\left[\mathbf{WW}^{\dagger}\right]=\sigma^{2}\mathbf{I}$
respectively. For any $\epsilon>0$, construct a random codebook $\mathcal{B}_{\mathbf{X}}=\left\{ \mathbf{X}_{1},\cdots,\mathbf{X}_{K}\right\} $
with $\frac{\log_{2}K}{n}\rightarrow R_{\epsilon}$ and $\frac{1}{n}\left\Vert \mathbf{X}_{k}\right\Vert ^{2}\in\left(1-2\epsilon,1-\epsilon\right)$
for all $k=1,\cdots,K$. 

Now suppose that a codeword $\mathbf{X}_{1}$ is transmitted. We consider
a receiver given by \[
\hat{\mathbf{X}}=\underset{\mathbf{X}\in\mathcal{B}_{\mathbf{X}}}{\arg\;\min}\; d_{c}^{2}\left(\mathcal{P}\left(\mathbf{X}\right),\mathcal{P}\left(\mathbf{Y}\right)\right),\]
where $\mathcal{P}\left(\mathbf{X}\right)\in\mathcal{G}_{n,1}\left(\mathbb{L}\right)$
and $\mathcal{P}\left(\mathbf{Y}\right)\in\mathcal{G}_{n,1}\left(\mathbb{L}\right)$
are planes generated by $\mathbf{X}$ and $\mathbf{Y}$ respectively.
It can be verified that \begin{align*}
\frac{\sigma^{2}}{1+\sigma^{2}-\epsilon}\le\underset{n\rightarrow\infty}{\lim}d_{c}^{2}\left(\mathcal{P}\left(\mathbf{X}_{1}\right),\mathcal{P}\left(\mathbf{Y}\right)\right) & \le\frac{\sigma^{2}}{1+\sigma^{2}-2\epsilon}.\end{align*}
By similar argument to the proof of Theorem \ref{thm:DRF_bounds},
if \[
R_{\epsilon}<\underset{n\rightarrow\infty}{\lim}\frac{\log_{2}K^{*}\left(\frac{\sigma^{2}}{1+\sigma^{2}-2\epsilon}\right)}{n}=\frac{\beta}{2}\log\left(1+\frac{1-2\epsilon}{\sigma^{2}}\right),\]
then $\Pr\left(\hat{\mathbf{X}}\neq\mathbf{X}_{1}\right)=\Pr\left(\exists j\ne1,\; d_{c}^{2}\left(\mathcal{P}\left(\mathbf{X}_{j}\right),\mathcal{P}\left(\mathbf{Y}\right)\right)\le\right.$
$\left.\frac{\sigma^{2}}{1+\sigma^{2}-2\epsilon}\right)\rightarrow0$.
Finally, let $\epsilon\rightarrow0$. The achievable error-free rate
for AWGN channel is then given by $\frac{\beta}{2}\log\left(1+\frac{1}{\sigma^{2}}\right)$,
which is the well known capacity of AWGN channel.

Therefore, transmission in an AWGN channel is essentially communication
on the Grassmann manifold: The decoder in the Grassmann manifold is
asymptotically optimal. Furthermore, based on the proof of Theorem
\ref{thm:DRF_bounds}, the capacity can be geometrically interpreted
as sphere packing in the Grassmann manifold.

\subsection{MIMO Communications with Beamforming Matrix Selection}

The Grassmann manifold also provides a useful analysis tool to MIMO
communications with finite rate feedback on beamforming matrix selection. 

Consider a MIMO systems with $L_{T}$ transmit antennas and $L_{R}$
receive antennas ($L_{R}<L_{T}$ is assumed). Suppose that the transmitter
sends $s$ ($s\le L_{T}$) independent data streams to the receiver.
Let $\mathcal{CN}\left(0,1\right)$ denote the symmetric complex Gaussian
distribution with zero mean and unit variance. Then the received signal
$\mathbf{Y}\in\mathbb{C}^{L_{R}\times1}$ is given by $\mathbf{Y}=\mathbf{HQX}+\mathbf{W}$,
where $\mathbf{H}\in\mathbb{C}^{L_{R}\times L_{T}}$ is the Rayleigh
fading channel state matrix with i.i.d. $\mathcal{CN}\left(0,1\right)$
entries, $\mathbf{Q}\in\mathbb{C}^{L_{T}\times s}$ is the beamforming
matrix satisfying $\mathbf{Q}^{\dagger}\mathbf{Q}=\mathbf{I}$, $\mathbf{X}\in\mathbb{C}^{s\times1}$
is the encoded Gaussian data source with zero mean and covariance
matrix $\frac{\rho}{s}\mathbf{I}$, and $\mathbf{W}\in\mathbb{C}^{L_{R}\times1}$
is the additive Gaussian noise with i.i.d. $\mathcal{CN}\left(0,1\right)$
entries. In our feedback model, we assume that only the receiver knows
channel state $\mathbf{H}$ perfectly. It will help the transmitter
choose a beamforming matrix through a finite rate feedback up to $R_{\mathrm{fb}}$
bits/channel realization. Specifically, A codebook of $\mathbf{Q}$,
say $\mathcal{B}_{\mathbf{Q}}$, satisfying $\left|\mathcal{B}_{\mathbf{Q}}\right|=2^{R_{\mathrm{fb}}}$
is declared to both the transmitter and the receiver. Given a channel
realization, the receiver selects a $\mathbf{Q}$ in $\mathcal{B}_{\mathbf{Q}}$
and feeds the corresponding index back to the transmitter.

The Grassmann manifold is related to throughput analysis of the above
system. Let $\mathbf{H}=\mathbf{U\Lambda V}^{\dagger}$ be the singular
value decomposition of $\mathbf{H}$ where$\mathbf{V}\in\mathbb{C}^{L_{T}\times L_{R}}$
satisfies $\mathbf{V}^{\dagger}\mathbf{V}=\mathbf{I}$. We consider
a suboptimal feedback function: for a given $\mathbf{H}$, the selected
beamforming matrix $\bar{\mathbf{Q}}\in\mathcal{B}_{\mathbf{Q}}$
is given by \[
\bar{\mathbf{Q}}=\underset{\mathbf{Q}\in\mathcal{B}_{\mathbf{Q}}}{\arg\;\min}\; d_{c}^{2}\left(\mathcal{P}\left(\mathbf{V}\right),\mathcal{P}\left(\mathbf{Q}\right)\right)\]
 where $\mathcal{P}\left(\mathbf{V}\right)\in\mathcal{G}_{L_{T},L_{R}}\left(\mathbb{C}\right)$
and $\mathcal{P}\left(\mathbf{Q}\right)\in\mathcal{G}_{L_{T},s}\left(\mathbb{C}\right)$
are planes generated by $\mathbf{V}$ and $\mathbf{Q}$ respectively.
Then the expected throughput $\mathcal{I}$ is upper bounded by \begin{align}
\mathcal{I} & \triangleq\mathrm{E}_{\mathbf{H}}\left[\log\left|\mathbf{I}+\frac{\rho}{s}\mathbf{H}\bar{\mathbf{Q}}\bar{\mathbf{Q}}^{\dagger}\mathbf{H}^{\dagger}\right|\right]\nonumber \\
 & \le L_{R}\cdot\log\left(1+\frac{\rho}{s}\frac{L_{T}}{L_{R}}\mathrm{E}_{\mathbf{V}}\left[\mathrm{tr}\left(\mathbf{V}^{\dagger}\bar{\mathbf{Q}}\bar{\mathbf{Q}}^{\dagger}\mathbf{V}\right)\right]\right).\label{eq:throughput-ubd}\end{align}
It is well known that the matrix $\mathbf{V}$ is isotropically distributed.
Hence, \[
\mathrm{E}_{\mathbf{V}}\left[\mathrm{tr}\left(\mathbf{V}^{\dagger}\bar{\mathbf{Q}}\bar{\mathbf{Q}}^{\dagger}\mathbf{V}\right)\right]=\min\left(s,L_{R}\right)-D\left(\mathcal{B}_{Q}\right),\]
 where $\mathcal{B}_{Q}=\left\{ \mathcal{P}\left(\mathbf{Q}\right):\;\mathbf{Q}\in\mathcal{B}_{\mathbf{Q}}\right\} $
is the codebook generated from $\mathcal{B}_{\mathbf{Q}}$. Based
on the distortion rate bounds (\ref{eq:DRF_bounds}), the bound (\ref{eq:throughput-ubd})
can be quantified for a given feedback rate $R_{\mathrm{fb}}$. 

It is noteworthy that beamforming matrix selection is essentially
unequal dimensional quantization when $s\ne L_{R}$. Similar models,
with minor modifications, have been adopted and explored in many papers.
The $s=L_{R}=1$ case has been studied in \cite{Sabharwal_IT03_Beamforming_MIMO,Love_IT03_Grassman_Beamforming_MIMO,Rao_SP07_Feedback_High_Resolution},
while \cite{Dai_05_Power_onoff_strategy_design_finite_rate_feedback}
discussed a more general equal dimensional quantization where $s\ge1$.
Recently,  \emph{unequal} dimensional quantization ($s=1,\; L_{R}>1$)
received attention for multi-user MIMO communications in \cite{Jindal_ISIT2006_Feedback_Reduction_MIMO_BC}.
Our model can be viewed as a generalization of all these works.

\section{Conclusion\label{sec:Conclusion}}

This paper considers unequal dimensional quantization on the Grassmann
manifold. An explicit volume formula for small balls is derived and
then the rate distortion tradeoff is accurately characterized. The
random codes are proved to be asymptotically optimal with probability
one. As applications of the derived results, a geometric model for
the capacity of AWGN channel is developed, and the effect of beamforming
matrix selection in MIMO systems is discussed.

\bibliographystyle{IEEEtran}
\bibliography{Bib/_love,Bib/_Rao,Bib/FeedbackMIMO_append,Bib/MIMO_basic,Bib/RandomMatrix,Bib/_Jindal,Bib/_Tse,Bib/_Heath,Bib/_Dai}

\begin{thebibliography}{10}
\providecommand{\url}[1]{#1}
\csname url@samestyle\endcsname
\providecommand{\newblock}{\relax}
\providecommand{\bibinfo}[2]{#2}
\providecommand{\BIBentrySTDinterwordspacing}{\spaceskip=0pt\relax}
\providecommand{\BIBentryALTinterwordstretchfactor}{4}
\providecommand{\BIBentryALTinterwordspacing}{\spaceskip=\fontdimen2\font plus
\BIBentryALTinterwordstretchfactor\fontdimen3\font minus
  \fontdimen4\font\relax}
\providecommand{\BIBforeignlanguage}[2]{{%
\expandafter\ifx\csname l@#1\endcsname\relax
\typeout{** WARNING: IEEEtran.bst: No hyphenation pattern has been}%
\typeout{** loaded for the language `#1'. Using the pattern for}%
\typeout{** the default language instead.}%
\else
\language=\csname l@#1\endcsname
\fi
#2}}
\providecommand{\BIBdecl}{\relax}
\BIBdecl

\bibitem{Dai_Globecom05_Quantization_bounds_Grassmann_manifold}
W.~Dai, Y.~Liu, and B.~Rider, ``Quantization bounds on {G}rassmann manifolds of
  arbitrary dimensions and {MIMO} communications with feedback,'' in \emph{IEEE
  Global Telecommunications Conference (GLOBECOM)}, 2005.

\bibitem{Tse_IT02_Communications_on_Grassmann_manifold}
L.~Zheng and D.~Tse, ``Communication on the {G}rassmann manifold: a geometric
  approach to the noncoherent multiple-antenna channel,'' \emph{IEEE Trans.
  Info. Theory}, vol.~48, no.~2, pp. 359--383, 2002.

\bibitem{Urbanke_IT01_Signal_Constellations}
D.~Agrawal, T.~J. Richardson, and R.~L. Urbanke, ``Multiple-antenna signal
  constellations for fading channels,'' \emph{IEEE Trans. Info. Theory},
  vol.~47, no.~6, pp. 2618 -- 2626, 2001.

\bibitem{Sabharwal_IT03_Beamforming_MIMO}
K.~K. Mukkavilli, A.~Sabharwal, E.~Erkip, and B.~Aazhang, ``On beamforming with
  finite rate feedback in multiple-antenna systems,'' \emph{IEEE Trans. Info.
  Theory}, vol.~49, no.~10, pp. 2562--2579, 2003.

\bibitem{Love_IT03_Grassman_Beamforming_MIMO}
D.~J. Love, J.~Heath, R.~W., and T.~Strohmer, ``Grassmannian beamforming for
  multiple-input multiple-output wireless systems,'' \emph{IEEE Trans. Info.
  Theory}, vol.~49, no.~10, pp. 2735--2747, 2003.

\bibitem{Dai_05_Power_onoff_strategy_design_finite_rate_feedback}
\BIBentryALTinterwordspacing
W.~Dai, Y.~Liu, V.~K.~N. Lau, and B.~Rider, ``On the information rate of {MIMO}
  systems with finite rate channel state feedback using beamforming and power
  on/off strategy,'' \emph{IEEE Trans. Info. Theory}, submitted, 2005.
  [Online]. Available: \url{http://arxiv.org/abs/cs/0603040}
\BIBentrySTDinterwordspacing

\bibitem{Rao_SP07_Feedback_High_Resolution}
J.~Zheng, E.~R. Duni, and B.~D. Rao, ``Analysis of multiple-antenna systems
  with finite-rate feedback using high-resolution quantization theory,''
  \emph{IEEE Trans. Signal Processing}, vol.~55, no.~4, pp. 1461--1476, 2007.

\bibitem{Jindal_ISIT2006_Feedback_Reduction_MIMO_BC}
N.~Jindal, ``A feedback reduction technique for mimo broadcast channels,'' in
  \emph{Proc. IEEE International Symposium on Information Theory (ISIT)}, 2006.

\bibitem{Conway_96_PackingLinesPlanes}
J.~H. Conway, R.~H. Hardin, and N.~J.~A. Sloane, ``Packing lines, planes, etc.,
  packing in {G}rassmannian spaces,'' \emph{Exper. Math.}, vol.~5, pp.
  139--159, 1996.

\bibitem{Barg_IT02_Bounds_Grassmann_Manifold}
A.~Barg and D.~Y. Nogin, ``Bounds on packings of spheres in the {G}rassmann
  manifold,'' \emph{IEEE Trans. Info. Theory}, vol.~48, no.~9, pp. 2450--2454,
  2002.

\bibitem{Henkel_IT2005_Sphere_Packing_Bounds_Grassmann_Stiefel}
O.~Henkel, ``Sphere-packing bounds in the {G}rassmann and {S}tiefel
  manifolds,'' \emph{IEEE Trans. Info. Theory}, vol.~51, no.~10, pp.
  3445--3456, 2005.

\bibitem{Han_IT2006_Unitary_Space_Time_Diversity_Analysis}
G.~Han and J.~Rosenthal, ``Unitary space-time constellation analysis: An upper
  bound for the diversity,'' \emph{IEEE Trans. Info. Theory}, vol.~52, no.~10,
  pp. 4713--4721, 2006.

\bibitem{James_54_Normal_Multivariate_Analysis_Orthogonal_Group}
A.~T. James, ``Normal multivariate analysis and the orthogonal group,''
  \emph{Ann. Math. Statist.}, vol.~25, no.~1, pp. 40 -- 75, 1954.

\bibitem{Adler_2004_Integrals_Grassmann}
M.~Adler and P.~van Moerbeke, ``Integrals over {Grassmannians} and random
  permutations,'' \emph{Advances in Mathematics}, vol. 181, no.~1, pp.
  190--249, 2004.

\bibitem{Dai_05_Quantization_Grassmannian_manifold}
\BIBentryALTinterwordspacing
W.~Dai, Y.~Liu, and B.~Rider, ``Quantization bounds on {G}rassmann manifolds
  and applications to {MIMO} systems,'' \emph{IEEE Trans. Info. Theory},
  Submitted, 2005. [Online]. Available: \url{http://arxiv.org/abs/cs/0603039}
\BIBentrySTDinterwordspacing

\bibitem{JanosGalambos1987_extreme_order_statistics}
J.~Galambos, \emph{The asymptotic theory of extreme order statistics},
  2nd~ed.\hskip 1em plus 0.5em minus 0.4em\relax Roberte E. Krieger Publishing
  Company, 1987.

\bibitem{Heath_ICASSP05_Quantization_Grassmann_Manifold}
B.~Mondal, R.~W.~H. Jr., and L.~W. Hanlen, ``Quantization on the {G}rassmann
  manifold: Applications to precoded {MIMO} wireless systems,'' in \emph{Proc.
  IEEE International Conference on Acoustics, Speech, and Signal Processing
  (ICASSP)}, 2005, pp. 1025--1028.

\end{thebibliography}

\end{document}